\documentclass{article}

\usepackage[preprint,nonatbib]{neurips_2020}

\usepackage[utf8]{inputenc} 
\usepackage[T1]{fontenc}    
\usepackage{hyperref}       
\usepackage{url}            
\usepackage{booktabs}       
\usepackage{amsfonts}       
\usepackage{nicefrac}       
\usepackage{microtype}      

\usepackage{multirow}
\usepackage{graphicx}
\usepackage{algorithm}
\usepackage{algorithmic}
\usepackage{makecell}
\usepackage{amsthm}
\usepackage{amsmath}

\newtheorem{definition}{Definition}
\newtheorem{theorem}{Theorem}
\newtheorem{corollary}{Corollary}
\newtheorem{lemma}{Lemma}

\title{A One-Pass Private Sketch for Most Machine Learning Tasks}

\author{Benjamin Coleman\\
  Department of Electrical and Computer Engineering\\
  Rice University\\
  Houston, TX 77005 \\
  \texttt{ben.coleman@rice.edu}\\
   \And
  Anshumali Shrivastava\\
  Department of Computer Science\\
  Rice University\\
  Houston, TX 77005 \\
  \texttt{anshumali@rice.edu}
}

\begin{document}

\maketitle

\begin{abstract}
Differential privacy (DP) is a compelling privacy definition that explains the privacy-utility tradeoff via formal, provable guarantees. Inspired by recent progress toward general-purpose data release algorithms, we propose a private sketch, or small summary of the dataset, that supports a multitude of machine learning tasks including regression, classification, density estimation, near-neighbor search, and more. Our sketch consists of randomized contingency tables that are indexed with locality-sensitive hashing and constructed with an efficient one-pass algorithm.  We prove competitive error bounds for DP kernel density estimation. Existing methods for DP kernel density estimation scale poorly, often exponentially slower with an increase in dimensions.
In contrast, our sketch can quickly run on large, high-dimensional datasets in a single pass.  Exhaustive experiments show that our generic sketch delivers a similar privacy-utility tradeoff when compared to existing DP methods at a fraction of the computation cost.  We expect that our sketch will enable differential privacy in distributed, large-scale machine learning settings. \end{abstract}

\section{Introduction}
\label{sec:Intro}
Large-scale data collection is an integral component of the modern machine learning pipeline. The success of a learning algorithm critically depends on the quantity and quality of the input data. Although vast amounts of data are available, the information is often of a personal or sensitive nature. Models can leak substantial information about individual participants, even if we only release predictions or outputs ~\cite{fredrikson2014privacy,fredrikson2015model,shokri2017membership}. Privacy is, therefore, an important design factor for machine learning algorithms. 

To protect against diverse and sophisticated attacks, $\epsilon$-differential privacy has emerged as a theoretically rigorous definition of privacy with robust guarantees~\cite{dwork2006differential}. Informally, an algorithm is differentially private if the inclusion (or exclusion) of any specific data record cannot substantially alter the output. In this paper, we consider the task of privately releasing a function $f_{\mathcal{D}}(q) \to \mathbb{R}$ that applies a pairwise operation $k(x,q)$ to each record $\{x_1, ... x_N\}$ in a dataset $\mathcal{D}$ and returns the sum. Although one could directly evaluate $f_{\mathcal{D}}(q)$ and release the result with the exponential mechanism, we can only evaluate $f_{\mathcal{D}}$ a finite number of times before the privacy budget runs out. The function release problem is to release a private summary $\mathcal{S}_{\mathcal{D}}$ of $f_{\mathcal{D}}$ that can answer an unlimited number of queries~\cite{hall2013differential}. We also seek error bounds for all queries\footnote{By ``all queries'', we mean all values of $q$. There is no polynomial algorithm for all general queries~\cite{ullman2013answering}.}, not just the global minimum as with empirical risk minimization~\cite{chaudhuri2011differentially} or a finite set of linear queries as in~\cite{hardt2012simple}. 

Many machine learning problems can be solved in the function release framework. When $k(x,q)$ is a loss function, we can train a model by minimizing $f_{\mathcal{D}}$. When $k(x,q)$ is a kernel function, $f_{\mathcal{D}}$ is the kernel density estimate - a convenient way to approximate the likelihood function in a classification model. As a result, function release has received considerable theoretical attention. There are elegant, general and powerful techniques to release essentially any function~\cite{hall2013differential,alda2017bernstein,formalmirshani19a,wang2013efficient}. However, function release has not yet been widely adopted in practice because existing methods fail to scale beyond small, low-dimensional datasets. The practical utility of function release is plagued by issues such as quadratic runtime and exponential memory requirements. For instance, many algorithms release $f_{\mathcal{D}}$ via function approximation over an interpolation lattice, but the size of the lattice grows exponentially with dimensions. 

\paragraph{Our Contribution:}

In this work, we propose a scalable approach to function release using the RACE sketch, a recent development in (non-private) data streaming algorithms. The RACE sketch can approximate pairwise kernel sums on streaming data for a specific class of kernels known as \textit{locality sensitive hash (LSH) kernels}. By restricting our attention to LSH kernel sums, we obtain fast streaming algorithms for private function release. RACE sketches consist of a small array of integers ($\sim$4 MB) and are well-suited to large-scale distributed settings. Private sketches are sought after by practitioners because they combine algorithmic efficiency with privacy guarantees~\cite{apple2017}. We show how to construct a private RACE sketch $\mathcal{S}_{\mathcal{D}}$ for the LSH kernel sum $f_{\mathcal{D}}$. We prove that our sketch is $\epsilon$-differentially private and derive pointwise error bounds for our approximation to $f_{\mathcal{D}}$. Our bounds are competitive with existing methods but come at a smaller computation cost. RACE easily scales to datasets with hundreds of dimensions and millions of entries. 

Although the restriction to LSH kernels might seem limiting, we argue that most machine learning tasks can be performed using the RACE sketch. We show how to express classification, linear regression, kernel density estimation (KDE), anomaly detection and mode finding in terms of specific LSH kernel compositions. We conduct an exhaustive set of experiments with KDE, classification and linear regression. Our experiments show that RACE can release useful functions for many machine learning tasks with a competitive privacy-utility tradeoff.

\section{Background}
We consider a dataset $\mathcal{D}$ of $N$ points in $\mathbb{R}^d$. Although our analysis naturally extends to any metric space, we restrict our attention to $\mathbb{R}^d$ for the sake of presentation. 





\subsection{Differential Privacy}
We use the well-established definition of differential privacy~\cite{dwork2006differential}.

\begin{definition}
\label{def:priv} \textbf{Differential Privacy}~\cite{dwork2006differential} A randomized function $\mathcal{A}$ is said to provide $(\epsilon,\beta)$-differential privacy if for all neighboring databases $\mathcal{D}$ and $\mathcal{D}'$ (which differ in at most one element) and all subsets $S$ in the codomain of $\mathcal{A}$,
$$\mathrm{Pr}[\mathcal{A}(\mathcal{D}) \in S] \leq e^{\epsilon}\mathrm{Pr}[A(\mathcal{D}')\in S] + \beta$$
\end{definition}
The parameter $\epsilon$ is the privacy budget. The privacy budget limits the amount of information that $\mathcal{A}(\mathcal{D})$ can leak about any individual element of $\mathcal{D}$. If $\beta > 0$, then $\mathcal{A}(\mathcal{D})$ might leak more information, but only with probability up to $\beta$. In this paper, we consider $\beta = 0$, which is simply called ``$\epsilon$-differential privacy.'' The Laplace mechanism~\cite{dwork2006differential} is a general method to satisfy $\epsilon$-differential privacy. By adding zero-mean Laplace noise to a real-valued function, we obtain differential privacy if the noise is scaled based on the sensitivity of the function (Definition~\ref{def:sens}). 
\begin{definition}
\label{def:sens}
\textbf{Sensitivity}~\cite{dwork2006differential} For a function $\mathcal{A}:\mathcal{D}\rightarrow \mathbb{R}^d$, the L1-sensitivity of $\mathcal{A}$ is
$$\Delta = \sup||\mathcal{A}(\mathcal{D}) - \mathcal{A}(\mathcal{D}')||_1$$
where the supremum is taken over all neighboring datasets $\mathcal{D}$ and $\mathcal{D}'$
\end{definition}
\begin{theorem}\textbf{Laplace Mechanism}~\cite{dwork2006differential} Let $\mathcal{A} : \mathcal{D} \rightarrow \mathbb{R}^d$ be a non-private function with sensitivity $\Delta$ and let $\mathbf{z}\sim \mathrm{Lap}(\Delta / \epsilon)^d$ ($d$-dimensional i.i.d Laplace vector). Then the function $\mathcal{A}(\mathcal{D}) + \mathbf{z}$ provides $\epsilon$-differential privacy.
\end{theorem}

\subsection{Related Work}
There are several techniques to release the kernel sum $f_{\mathcal{D}}$ with differential privacy. A common approach is to decompose $f_{\mathcal{D}}$ into a set of weighted basis functions~\cite{wasserman2010statistical}. We truncate the basis expansion to $m$ terms and represent $f_{\mathcal{D}}$ as a set of weights in $\mathbb{R}^m$. The weights are made private via the Laplace mechanism and used to release a private version of $f_{\mathcal{D}}$. Each basis term in the representation increases the quality of the approximation but degrades the quality of the weights, since the privacy budget $\epsilon$ is shared among the $m$ weights. This is a bias-variance tradeoff: we trade variance in the form of Laplace noise with bias in the form of truncation error. The Fourier basis~\cite{hall2013differential}, Bernstein basis~\cite{alda2017bernstein}, trigonometric polynomial basis~\cite{wang2013efficient} and various kernel bases have all been used for private function release. Such methods are most effective when $f_{\mathcal{D}}$ is a smooth function. 

An alternative set of techniques rely on functional data analysis ~\cite{formalmirshani19a} or synthetic databases ~\cite{hardt2012simple,balog2018differentially}. In~\cite{formalmirshani19a}, the authors use densities over function spaces to release a smoothed approximation to $f_{\mathcal{D}}$. The main idea of ~\cite{hardt2012simple} and ~\cite{balog2018differentially} is to release a set of weighted synthetic points that can be used to estimate $f_{\mathcal{D}}$. Table~\ref{table:relatedwork} compares these methods based on approximation error and computation. 

\begin{table}[t]
\begin{center}
\begin{tabular}{ l | l | l | l }
\toprule
Method & Error Bound & Runtime & Comments \\
\midrule
\makecell[l]{Bernstein\\polynomials~\cite{alda2017bernstein}} & $d^{\frac{d}{d+H}}\left(\frac{1}{\epsilon}\log \frac{1}{\delta} \right)^{\frac{H}{d+H}}$ & $O(dNM^d)$& \makecell[l]{$M\geq 2$. Memory is also\\ exponential in $d$.} \\
\midrule
\makecell[l]{PFDA~\cite{formalmirshani19a}} & 
$\frac{2}{\epsilon}\sqrt{\log \frac{2}{\beta} \log \frac{1}{\delta} \frac{C}{\phi}}$
& $O(dN^2)$ & \makecell[l]{$C$ and $\phi$ are task-dependent\\ $(\epsilon,\beta)$-differential privacy } \\
\midrule
\makecell[l]{MWEM~\cite{hardt2012simple}} & $N^{\frac{2}{3}}\left(\frac{\log N \log |Q|}{\epsilon}\right)^{1/3}$ & $O(dN|Q|)$ &  \makecell[l]{$Q$ is a set of query points. Holds \\ with probability $1 - 1 / \mathrm{poly}(|Q|)$ }\\
\midrule
\makecell[l]{Trigonometric\\polynomials~\cite{wang2013efficient}} & $\frac{1}{\epsilon}N^{\frac{2d}{2d+H}}$ & $O(dN^{1 + \frac{d}{2d+H}})$ & \makecell[l]{The result holds with probability\\ $1 - \delta$ for $\delta\geq10e^{-\frac{1}{5} N^{d / 2d + K}}$ }\\
\midrule
\makecell[l]{This work} & $\left(\frac{N}{\epsilon}\log \frac{1}{\delta}\right)^{\frac{1}{2}}$ & $O(dN)$ & \makecell[l]{
Applies only for LSH kernels. \\ Efficient streaming algorithm.}\\
\bottomrule
\end{tabular}
\caption{Summary of related methods to release the kernel sum $f_\mathcal{D} = \sum_\mathcal{D} k(\mathbf{x},\mathbf{q})$ for an $N$ point dataset $\mathcal{D}$ in $\mathbb{R}^d$. Unless otherwise stated, the error is attained with probability $1 - \delta$ and $\epsilon$-differential privacy. We hide constant factors and adjust results to estimate $f_\mathcal{D}$ rather than the KDE ($N^{-1}f_\mathcal{D}(\mathbf{q})$) when necessary. $H$ is a kernel smoothness parameter.}
\label{table:relatedwork}
\end{center}
\vspace{-0.5cm}
\end{table}

\subsection{Locality-Sensitive Hashing}
\label{sec:LSH}
\paragraph{LSH Functions:} An LSH family $\mathcal{F}$ is a family of functions $l(\mathbf{x}): \mathbb{R}^{d}\to \mathbb{Z}$ with the following property: Under $l(\mathbf{x})$, similar points have a high probability of having the same hash value. We say that a collision occurs whenever two points have the same hash code, i.e. $l(\mathbf{x}) = l(\mathbf{y})$. In this paper, we use a slightly different definition than the original~\cite{indyk1998approximate} because we require the collision probability at all points. 

\begin{definition}
\label{def:lsh}
We say that a hash family $\mathcal{F}$ is locality-sensitive with collision probability $k(\cdot,\cdot)$ if for any two points $\mathbf{x}$ and $\mathbf{y}$ in $\mathbb{R}^d$, $l(\mathbf{x}) = l(\mathbf{y})$ with probability $k(\mathbf{x},\mathbf{y})$
under a uniform random selection of $l(\cdot)$ from $\mathcal{F}$. 
\end{definition}

\paragraph{LSH Kernels:} When the collision probability $k(\mathbf{x},\mathbf{y})$ is a monotone decreasing function of the distance metric $\mathrm{dist}(\mathbf{x},\mathbf{y})$, one can show that $k$ is a positive semidefinite kernel function~\cite{coleman2020race}. We say that a kernel function $k(\mathbf{x},\mathbf{y})$ is an \textit{LSH kernel} if it forms the collision probability for an LSH family. For a kernel to be an LSH kernel, it must obey the conditions described in~\cite{chierichetti2012preserving}. A number of well-known LSH families induce useful kernels~\cite{gionis1999similarity}. For example, there are LSH kernels that closely resemble the cosine, Laplace and multivariate Student kernels~\cite{coleman2020race}. 

\subsection{RACE Sketch}
LSH kernels are interesting because there are efficient algorithms to estimate the quantity
\begin{equation}
    f_\mathcal{D}(\mathbf{q}) = \sum_{\mathbf{x}\in\mathcal{D}} k(\mathbf{x},\mathbf{q})
\end{equation}
when $k(\mathbf{x},\mathbf{q})$ is an LSH kernel. In~\cite{coleman2020race}, the authors present a one-pass streaming algorithm to estimate $f_\mathcal{D}$. The algorithm produces a RACE (Repeated Array of Count Estimators) sketch $\mathcal{S}_\mathcal{D}\in \mathbb{Z}^{R\times W}$, a 2D array of integers that we index using LSH functions. This array is sufficient to report $f_\mathcal{D}(\mathbf{q})$ for any query $\mathbf{q}\in\mathbb{R}^d$. We begin by constructing $R$ functions $\{l_1(\mathbf{x}),...l_R(\mathbf{x})\}$ from an LSH family $\mathcal{F}$ with the desired collision probability. When an element $\mathbf{x}$ arrives from the stream, we hash $\mathbf{x}$ to get $R$ hash values, one for each row of $\mathcal{S}_\mathcal{D}$. We increment row $i$ at location $l_i(\mathbf{x})$ and repeat for all elements in the dataset. To approximate $f_\mathcal{D}(\mathbf{q})$, we return the mean of $\mathcal{S}_\mathcal{D}[r,l_r(\mathbf{q})]$ over the $R$ rows. 


\paragraph{RACE for Density Estimation:} This streaming algorithm approximates the kernel density estimate (KDE) for all possible queries in a single pass. RACE is also a mergable summary with an efficient distributed implementation. In practice, RACE estimates the KDE with 1\% error using a 4 MB array, even for large high-dimensional datasets. To prove rigorous error bounds, observe that each row of $\mathcal{S}_\mathcal{D}$ is an unbiased estimator of $f_\mathcal{D}$. The main theoretical result of~\cite{coleman2020race} is stated below as Theorem~\ref{thm:race}. 

\begin{theorem}
\label{thm:race}
\textbf{Unbiased RACE Estimator}\cite{coleman2020race} Suppose that $X$ is the query result for one of the rows of $\mathcal{S}_\mathcal{D}$. That is, $X= \mathcal{S}_\mathcal{D}[r,l_r(\mathbf{q})]$

\noindent\begin{minipage}{0.3\textwidth}
$$ \mathbb{E}[X] = f_\mathcal{D}(\mathbf{q})$$
\end{minipage}
\begin{minipage}{0.7\textwidth}
$$ \mathrm{var}(X)\leq \left(\tilde{f}_\mathcal{D}(\mathbf{q})\right)^2 = \left(\sum_{\mathbf{x}\in\mathcal{D}}\sqrt{k(\mathbf{x},\mathbf{q})}\right)^2$$
\end{minipage}

\end{theorem}

\paragraph{RACE for Empirical Risk Minimization:} Theorem~\ref{thm:race} holds for all collision probabilities, even those that are neither continuous nor positive semidefinite. Thus, we can use RACE to approximate the empirical risk for a variety of losses. Suppose we are given a dataset $\mathcal{D} = \{\mathbf{z}_1...\mathbf{z}_N\}$ of training examples and a loss function $L(\theta, \mathbf{z})$, where $\theta$ is a parameter that describes the model. Empirical Risk Minimization (ERM) consists of finding a parameter $\theta$ (and thus a model) that minimizes the mean loss over the training set. 
RACE can approximate the empirical risk when $L(\theta, \mathbf{z})$ can be expressed in terms of LSH collision probabilities. Although we cannot analytically find the gradient of the RACE count values, derivative-free optimization~\cite{conn2009introduction} is highly effective for RACE sketches. With the right LSH family, RACE can perform many machine learning tasks in the one-pass streaming setting. 

\section{Private Sketches with RACE}


We propose a private version of the RACE sketch. We obtain $\epsilon$-differential privacy by applying the Laplace mechanism to each count in the RACE sketch array. Algorithm~\ref{alg:sketch} introduces a differentially private method to release the RACE sketch, illustrated in Figure~\ref{fig:privateRACE}. It is straightforward to see that Algorithm~\ref{alg:sketch} only requires $O(NR)$ hash computations. Assuming fixed $R$, we have $O(dN)$ runtime.

\begin{minipage}{0.46\textwidth}
\begin{algorithm}[H]
\begin{algorithmic}
\STATE {\bf Input:} Dataset $\mathcal{D}$, privacy budget $\epsilon$, LSH family $\mathcal{F}$, dimensions $R\times W$
\STATE {\bf Output:} Private sketch $\mathcal{S}_{\mathcal{D}} \in \mathbb{Z}^{R\times W}$
\STATE {\bf Initialize:} $R$ independent LSH functions $\{l_1, ..., l_R\}$ from the LSH family $\mathcal{F}$
\STATE $\mathcal{S}_{\mathcal{D}} \leftarrow \mathbf{0}^{R\times W}$
\FOR{$\mathbf{x} \in \mathcal{D}$}
    \FOR {$r$ in $1$ to $R$}
        \STATE Increment $\mathcal{S}_{\mathcal{D}}[r,l_r(\mathbf{x})]$
    \ENDFOR
\ENDFOR
\STATE $\mathcal{S}_{\mathcal{D}} = \left \lfloor{\mathcal{S}_{\mathcal{D}} + Z}\right \rfloor$ where $Z \overset{\text{iid}}{\sim} \mathrm{Lap}(R\epsilon^{-1})$
 \end{algorithmic}
  \caption{Private RACE sketch}
 \label{alg:sketch}
\end{algorithm}
\end{minipage}
\begin{minipage}{0.46\textwidth}
\begin{algorithm}[H]
\begin{algorithmic}
\STATE {\bf Input:} Sketch $\mathcal{S}_{\mathcal{D}}$, query $\mathbf{q}$, the same $R$ LSH functions from Algorithm~\ref{alg:sketch}
\STATE {\bf Output:} Estimate of $N^{-1}f_{\mathcal{D}}(\mathbf{q})$
\STATE $\widehat{N} \leftarrow R^{-1} \sum_{i,j} \mathcal{S}_{\mathcal{D}}[i,j]$
\STATE $\hat{f}_{\mathcal{D}} \leftarrow 0$
\FOR {$r$ in $1$ to $R$}
    \STATE $\hat{f}_{\mathcal{D}} = \hat{f}_{\mathcal{D}} + 
    \frac{1}{R} \mathcal{S}_{\mathcal{D}}[r,l_r(\mathbf{q})]$
\ENDFOR
\STATE \textbf{Return:} $\hat{f}_{\mathcal{D}}/ \widehat{N}$
\end{algorithmic}
\caption{RACE query}
 \label{alg:query}
\end{algorithm}
\vspace{0.55cm}
\end{minipage}

\subsection{Privacy}
For the purposes of Definition~\ref{def:priv}, we consider the function $\mathcal{A}: \mathcal{D} \rightarrow \mathbb{R}^{R\times W}$ to be Algorithm~\ref{alg:sketch}. The codomain of $\mathcal{A}$ is the set of all RACE sketches with $R$ rows and $W$ columns. Our main theorem is that the value returned by $\mathcal{A}(\mathcal{D})$ is $\epsilon$-differentially private. That is, our sketch is differentially private. We begin the proof by applying the parallel composition theorem~\cite{dwork2014algorithmic} to the counters in one row, which each see a disjoint partition of $\mathcal{D}$. Then, we apply the sequential composition theorem~\cite{dwork2014algorithmic} to the set of $R$ rows. Due to space limitations, we defer the full proof to the Appendix. 


\begin{theorem}
For any $R > 0$, $W > 0$, and LSH family $\mathcal{F}$, the output of Algorithm~\ref{alg:sketch}, or the RACE sketch $\mathcal{S}_{\mathcal{D}}$, is $\epsilon$-differentially private. 
\end{theorem}


\subsection{Utility}


Since one can construct many learning algorithms from a sufficiently good estimate of $f_{\mathcal{D}}$~\cite{alda2017bernstein}, we focus on utility guarantees for the RACE estimate of $f_{\mathcal{D}}(\mathbf{q})$. Since RACE is a collection of unbiased estimators for $f_{\mathcal{D}}(\mathbf{q})$, our proof strategy is to bound the distance between the RACE estimate and the mean. To bound the variance of the private RACE estimator, we add the independent Laplace noise variance $2R^2\epsilon^{-2}$ to the bound from Theorem~\ref{thm:race}. Theorem~\ref{thm:tradeoff} follows using the median-of-means procedure.

\begin{theorem}
\label{thm:tradeoff}
Let $\hat{f}_{\mathcal{D}}(\mathbf{q})$ be the median-of-means estimate using an $\epsilon$-differentially private RACE sketch with $R$ rows and $\tilde{f}_{\mathcal{D}}(\mathbf{q}) = \sum_{\mathcal{D}} \sqrt{k(\mathbf{x},\mathbf{q})}$. Then with probability $1 - \delta$, 
$$|\hat{f}_{\mathcal{D}}(\mathbf{q}) - f_{\mathcal{D}}(\mathbf{q})| \leq \left(\frac{\tilde{f}^2_\mathcal{D}(\mathbf{q})}{R} + \frac{2}{\epsilon^2}R \right)^{1/2}\sqrt{32 \log 1/\delta}$$

\end{theorem}

Theorem~\ref{thm:tradeoff} suggests a tradeoff for which there is an optimal value of $R$. If we increase the number of rows $R$ in the sketch, we improve the estimator but must add more Laplace noise. To get our main utility guarantee, we choose an optimal $R$ that minimizes the error bound.

\begin{corollary}
\label{cor:util}
Put $R = \lceil\frac{1}{\sqrt{2}}\tilde{f}_\mathcal{D}(\mathbf{q})\epsilon\rceil$. Then the approximation error bound is
$$|\hat{f}_{\mathcal{D}}(\mathbf{q}) - f_{\mathcal{D}}(\mathbf{q})| \leq 16\left(\frac{\tilde{f}_\mathcal{D}(\mathbf{q})}{\epsilon}\log 1/\delta\right)^{1/2} \leq 16\sqrt{\frac{N}{\epsilon}\log 1/\delta}$$
\end{corollary}

If we divide both sides of Corollary~\ref{cor:util} by $f_{\mathcal{D}}$, we bound the relative (or percent) error rather than the absolute error. Corollary~\ref{cor:util} suggests that it is hardest to achieve a small relative error when $f_{\mathcal{D}}$ is small. This agrees with our intuition about how the KDE should behave under differential privacy guarantees. Fewer individuals make heavy contributions to $f_{\mathcal{D}}$ in low-density regions than in high-density ones, so the effect of the noise is worse. 


\begin{figure*}[t]
\centering
\includegraphics[width=\textwidth,keepaspectratio]{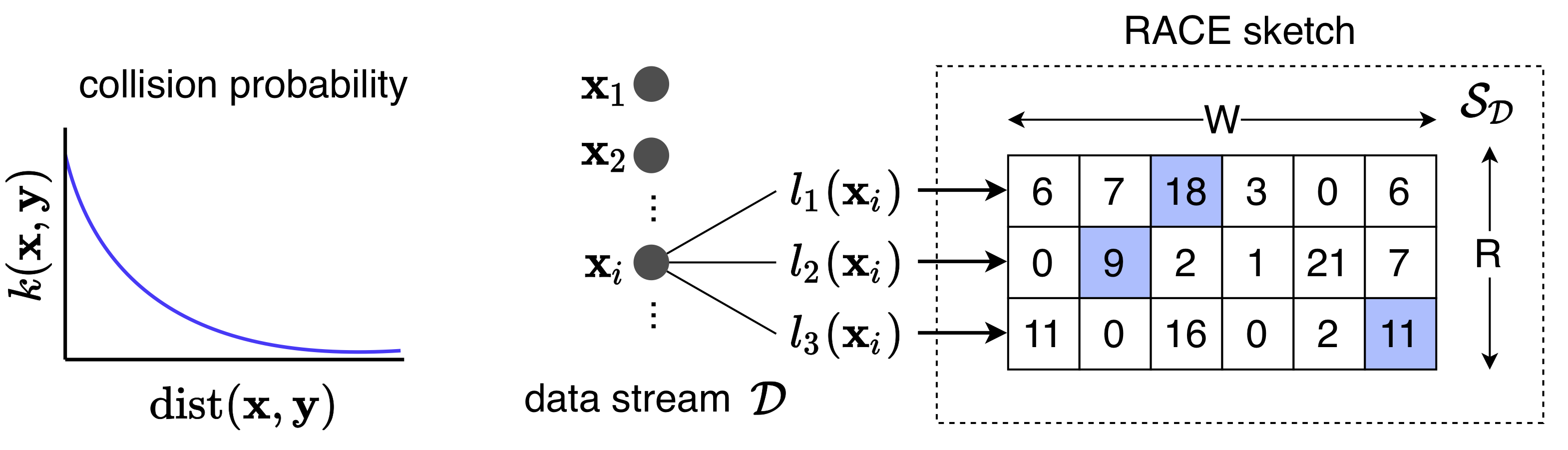}
\vspace{-0.5cm}
\caption{Illustration of Algorithm~\ref{alg:sketch} for $\mathcal{S}_{\mathcal{D}}\in \mathbb{Z}^{3\times6}$. We hash each element in the stream with LSH functions $\{l_1,l_2,l_3\} \in \mathcal{F}$ having collision probability $k(\mathbf{x},\mathbf{y})$. In this example, $l_1(\mathbf{x}_i) = 3$, $l_2(\mathbf{x}_i) = 2$ and $l_3(\mathbf{x}_i) = 6$. We increment the highlighted cells. The addition of the Laplace noise is not shown in the figure, but is done by perturbing each count in $\mathcal{S}_{\mathcal{D}}$.}
\label{fig:privateRACE}
\vspace{-0.2cm}
\end{figure*}

\section{Applications}

Because RACE can release LSH kernel sums, our sketch is broadly useful for many algorithms in machine learning. In particular, we discuss private density estimation, classification, regression, mode finding, anomaly detection and diversity sampling using RACE. 

\textbf{Kernel Density Estimation:} To use RACE for KDE, we select one or more LSH kernels and construct sketches with Algorithm~\ref{alg:sketch}. We require one RACE sketch for each kernel and bandwidth setting, and we return the result of Algorithm~\ref{alg:query}. Note that the query no longer has access to $N$, the number of elements in the private dataset. Therefore, we estimate $N$ directly from the private sketch.

\textbf{Mode Finding:} One can apply gradient-free optimization to KDE to recover the modes of the data distribution. This works surprisingly well, but in general the KDE is a non-convex function. One can also apply linear programming techniques that use information about the hash function partitions to identify a point in a partition with the highest count values.

\textbf{Naive Bayes Classification:} Using kernel density classification, a well-developed result from statistical hypothesis testing~\cite{john1995estimating}, we can construct classifiers with RACE under both the maximum-likelihood and maximum a posteriori (MAP) decision rules. Suppose we are given a training set $\mathcal{D}$ with $M$ classes $C_1,... C_M$ and a query $\mathbf{q}$. We can represent the empirical likelihood $\mathrm{Pr}[\mathbf{q}|C_i,\mathcal{D}]$ with a sketch of the KDE for class $i$. Algorithm~\ref{alg:query} returns an estimate of this probability, which may be used directly by a naive Bayes classifier or other type of probabilistic learner. 

\textbf{Anomaly Detection / Sampling:} Anomaly detection can be cast as a KDE classification problem. If Algorithm~\ref{alg:query} reports a low density for $\mathbf{q}$, then the training set contains few elements similar to $\mathbf{q}$ and thus $\mathbf{q}$ is an outlier. This principle is behind the algorithms in~\cite{luo2018arrays} and~\cite{coleman2019diversified}.


\textbf{Linear Regression:} If we use an asymmetric LSH family~\cite{shrivastava2014asymmetric}, we can construct a valid surrogate loss for the linear regression loss $(\langle[\mathbf{x},y],[\theta,-1]\rangle)^2$. The key insight is that the signed random projection (SRP) LSH kernel is monotone in the inner product. If we apply SRP to both $[\mathbf{x},y]$ and $-[\mathbf{x},y]$, we obtain an LSH kernel with two components. One component is monotone increasing with the inner product, while the other is monotone decreasing. The resulting surrogate loss is a monotone function of the L2 loss and a convex function of $\theta$. 



\section{Experiments}

\begin{table*}[t]
  \centering
  \begin{tabular}{ l|c|c|c|l|l  } 
\toprule
Dataset & $N$ & $d$ & $\sigma$ & Description & Task\\
\midrule
NYC & 25k & 1 & 5k & NYC salaries (2018) & \multirow{4}[1]{*}{KDE} \\
SF & 29k & 1 & 5k & SF salaries (2018) &  \\
skin & 241k & 3 & 5.0 & RGB skin tones & \\
codrna & 57k & 8 & 0.5 & RNA genomic data & \\
\midrule
nomao & 34k & 26 & 0.6 & User location data & \multirow{3}[1]{*}{Classification}\\
occupancy & 17k & 5 & 0.5 & Building occupancy & \\
pulsar & 17k & 8 & 0.1 & Pulsar star data & \\
\midrule
airfoil & 1.4k & 9 & - & Airfoil parameters and sound level & \multirow{3}[1]{*}{Regression}\\
naval & 11k & 16 & - & Frigate turbine propulsion &  \\
gas & 3.6k & 128 & - & Gas sensor, different concentrations & \\
\bottomrule
\end{tabular}
\caption{Datasets used for KDE and classification experiments. Each dataset has $N$ entries with $d$ features. $\sigma$ is the kernel bandwidth.}
\label{tab:datasets}
\vspace{-0.2cm}
\end{table*}

\begin{table}[t]
\centering
\begin{tabular}{c|c|c|c|c}
\toprule
& PFDA & Bernstein & KME & RACE \\
\midrule
Preprocess & $> 3$ days & 2.3 days & 6 hr & 13 min \\
Query & - & 6.2 ms & 1.2 ms & 0.4 ms \\
\bottomrule
\end{tabular}
\vspace{0.1cm}
\caption{Computation time for KDE on the skin dataset. Note that we were unable to run PFDA on this dataset. Instead, we report the PFDA runtime for the (smaller) NYC dataset.}
\label{tab:computation}
\vspace{-0.5cm}
\end{table}

We perform an exhaustive comparison for KDE, classification and regression. For KDE, we estimate the density of the salaries for New York City (NYC) and San Francisco (SF) city employees in 2018, as well as the high-dimensional densities for the skin and codrna UCI datasets. We use the same LSH kernel as the authors of~\cite{coleman2020race}. We use UCI datasets for the regression and classification experiments. Most private KDE methods require the data to lie in the unit sphere or cube; we scale the inputs accordingly. Table~\ref{tab:datasets} presents the datasets used in our experiments. We make the following considerations for our baseline methods.




\textbf{KDE:} We implemented several function release methods from Table~\ref{table:relatedwork}. We also compare against the Fourier basis estimator~\cite{hall2013differential,wasserman2010statistical} and the kernel mean embedding (KME) method~\cite{balog2018differentially}. We use the Python code released by~\cite{balog2018differentially} and implemented the other methods in Python. To give a fair comparison, we profiled and optimized each of our baselines. We were unable to run Fourier and PFDA in more than one dimension because they both require function evaluations on a Bravais lattice, which scales exponentially. We show the density estimates and their errors in Figure~\ref{fig:KDEResults}. 

\textbf{Classification:} We compare a maximum-likelihood RACE classifier against a regularized logistic regression classifier trained using objective perturbation~\cite{chaudhuri2011differentially}. We average over the Laplace noise and report the accuracy on a held-out test set in Figure~\ref{fig:ClassResults}.


\textbf{Regression:} We compare RACE regression against five algorithms: sufficient statistics perturbation (SSP), objective perturbation (ObjPert), posterior sampling (OPS), and adaptive versions of SSP and OPS (AdaSSP and AdaOPS~\cite{wang2018revisiting}). We use the MatLab code released by~\cite{wang2018revisiting} and report the mean squared error (MSE) on a held-out test set in Figure~\ref{fig:RegressionResults}.

\begin{figure*}[t]
\centering
  \includegraphics[width=\textwidth,keepaspectratio]{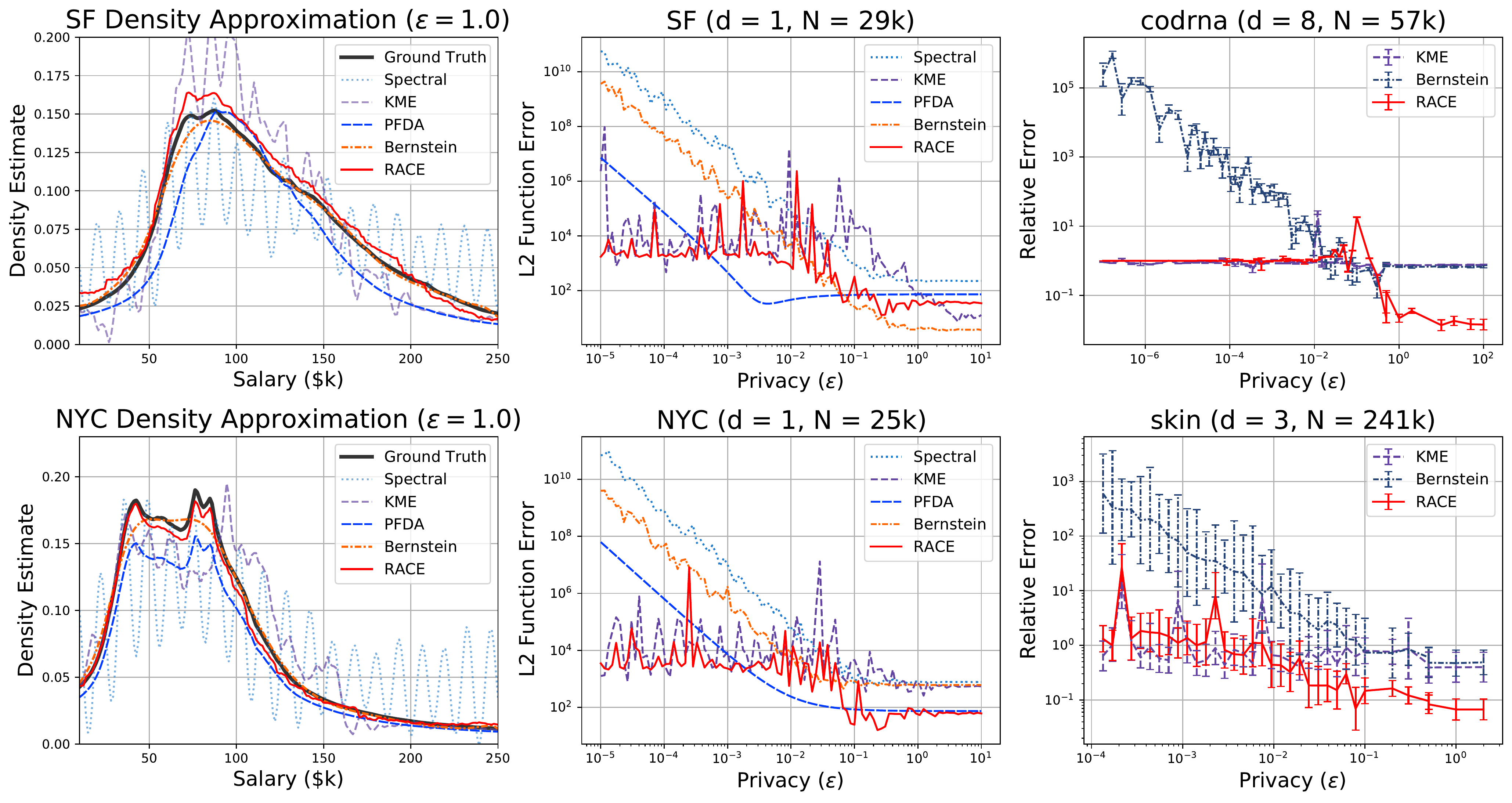}
\vspace{-0.5cm}
\caption{Privacy-utility tradeoff for private function release methods. We report the L2 function error and the mean relative error for 2000 held-out queries.}
\label{fig:KDEResults}
\vspace{-0.2cm}
\end{figure*}





\textbf{Computation: } Table~\ref{tab:computation} displays the computation time needed to construct a useful function release. Note that Bernstein release can run faster if we use fewer interpolation points (see Table~\ref{table:relatedwork}), but we still required at least 12 hours of computation for competitive results. The Bernstein mechanism requires many binomial coefficient evaluations, which were expensive even when we used optimized C code. KME requires a large-scale kernel matrix computation, and PFDA required several days for an expensive eigenvalue computation. RACE construction time varies based on $R$, but is substantially faster than all baselines.

\begin{figure*}[t]
\centering
  \includegraphics[width=\textwidth,keepaspectratio]{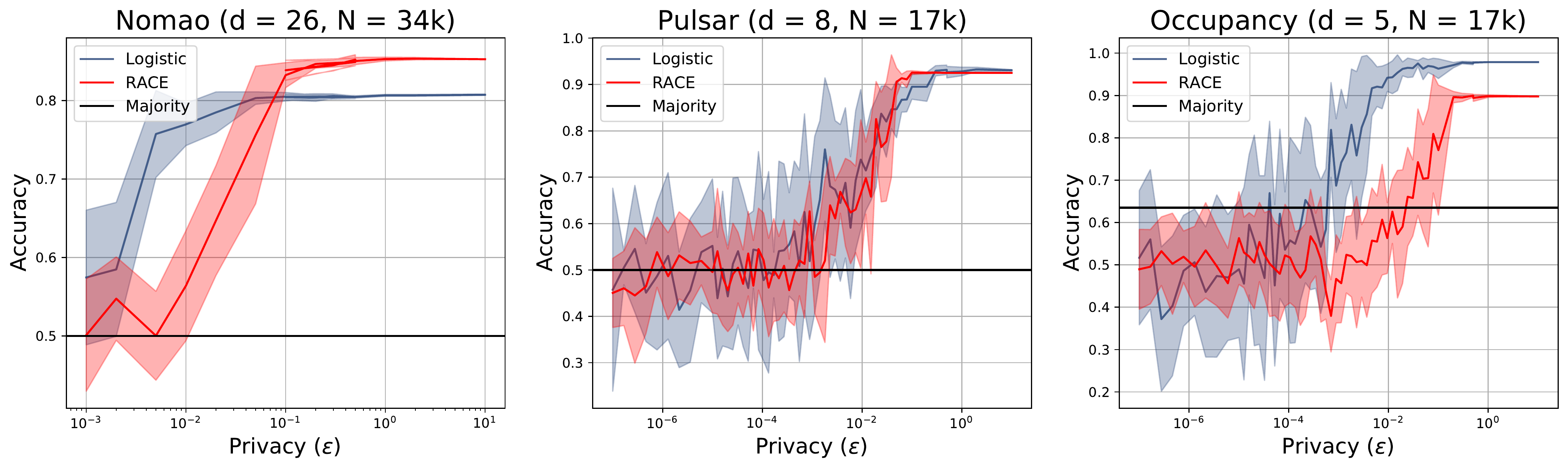}
\vspace{-0.5cm}
\caption{Binary classification experiments. We show the privacy-utility tradeoff for a private logistic regression classifier and the RACE max-likelihood classifier. Average over 10 repetitions. }
\label{fig:ClassResults}
\end{figure*}
\begin{figure*}[t]
\centering
  \includegraphics[width=\textwidth,keepaspectratio]{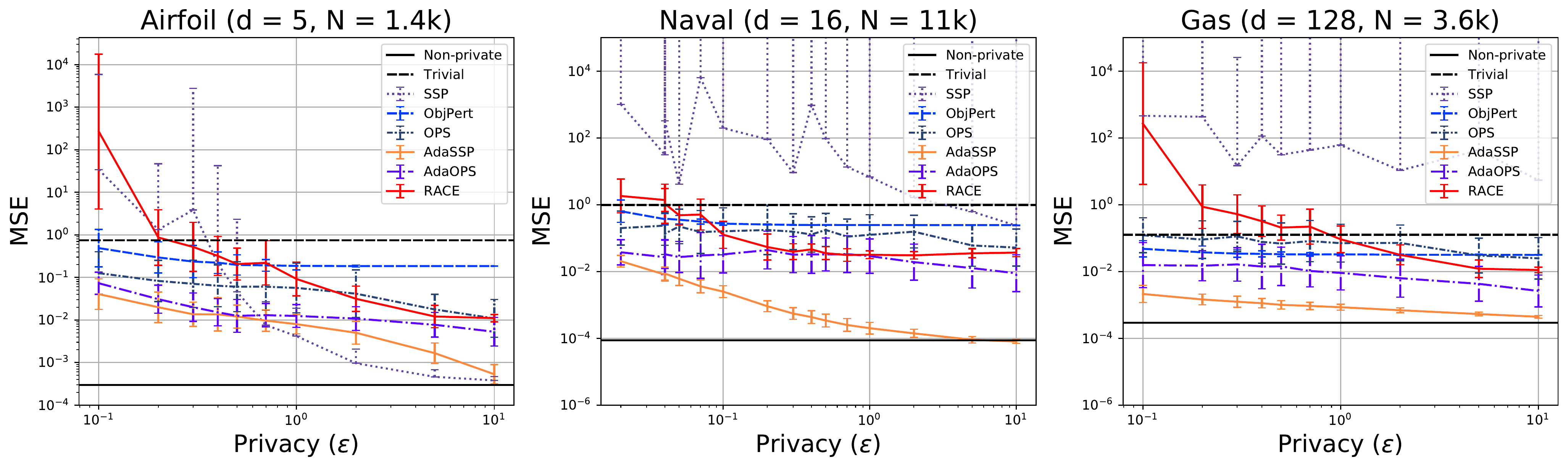}
\vspace{-0.5cm}
\caption{Linear regression experiments. We show the privacy-utility tradeoff for RACE and several other linear regression methods. Average over 10 repetitions. }
\label{fig:RegressionResults}
\vspace{-0.5cm}
\end{figure*}




\section{Discussion}

\paragraph{Function Release at Scale:}
RACE is ideal for private function release in large-scale distributed applications. Although PFDA and the Bernstein mechanism have the strongest error bounds, they required days to produce the experimental results in Figure~\ref{fig:KDEResults}. This is a serious barrier in practice - it would require $> 2^{128} \approx 10^{38}$ computations to run the Bernstein mechanism on the UCI gas dataset. The RACE sketch has a small memory footprint, inexpensive streaming updates and can quickly run on high-dimensional datasets. We believe that RACE can make differentially private function release a viable tool for real-world applications. 



Our sketch is also convenient to use and deploy in a production environment because it is relatively insensitive to hyperparameters. In general, we found that density estimation requires a larger sketch with more rows $R$ than classification or regression. Classification and regression problems benefit from a smaller hash range $W$ than function release. Higher dimensional problems also require a larger hash range. However, these choices are not critical and any $R \in [100,1\text{k}]$ with $W \in [100,1\text{k}]$ will provide good results. Hyperparameter tuning can be accomplished using very little of the overall privacy budget.





\paragraph{Privacy and Utility:}
Our sketch beats interpolation-based methods for private function release when the dataset has more than a few dimensions and when $f_{\mathcal{D}}$ is not smooth. Depending on smoothness, our $\sqrt{N\epsilon^{-1}}$ error rate improves upon the rates of~\cite{wang2013efficient} and~\cite{hardt2012simple} for LSH kernels. In our experiments (Figure~\ref{fig:KDEResults}), the Bernstein mechanism outperforms RACE on the SF dataset but fails to capture the nuances of the NYC salary distribution, which has sharp peaks. RACE preserves the details of $f_{\mathcal{D}}$ because the Laplace noise can only make local changes to the RACE structure. Suppose we generate a Laplace noise that is unusually large for one particular RACE counter. Queries that do not involve the problematic counter are unaffected. If we perturb one or two of the most important weights of a series estimator, the changes propagate to all queries. 

Although RACE can only estimate LSH kernels, the space of LSH kernels is sufficiently large that we can release many useful functions. Our RACE surrogate loss outperformed objective perturbation~\cite{chaudhuri2011differentially} and compared favorably to other linear regression baselines. RACE also performed well in our classification experiments for $\epsilon > 10^{-1}$, providing a competitive utility tradeoff for practical privacy budgets. Our experiments show that RACE can privately release useful functions for many machine learning problems.

\section{Conclusion}
We have presented RACE, a differentially private sketch that can be used for a variety of machine learning tasks. RACE is competitive with the state of the art on tasks including density estimation, classification, and linear regression. RACE sketches can be constructed in the one-pass streaming setting and are highly computationally efficient. At the same time, our sketch offers good performance on many machine learning tasks and is reasonably efficient in terms of the privacy budget. Given the utility, simplicity and speed of the algorithm, we expect that RACE will enable differentially private machine learning in large-scale settings.

\section*{Broader Impact}

Privacy is a serious concern in machine learning and a barrier to the widespread adoption of useful technology. For example, machine learning models can substantially improve medical treatments, inform social policy, and help with financial decision-making. While such applications are overall beneficial to society, they often rely on sensitive personal data that users may not want to disclose. For instance, a 2019 survey found that a large majority ($\sim$ 80\%) of Americans are concerned with but do not understand how technology companies use the data they collect~\cite{pew2019}. Data management practices have recently come under increased scrutiny, especially after the enactment of the European GDPR Act in 2018. To improve public confidence in their services, Google~\cite{erlingsson2014rappor} and Apple~\cite{apple2017} both advertise differential privacy as a central part of their data collection process. 

Differential privacy provides strong protection against malicious data use but is difficult to apply in practice. This is particularly true for functional data release, which is often prohibitively expensive in terms of computation cost. Our work substantially reduces the computational cost of this useful technique to achieve differential privacy. By reducing the computational requirement, we make privacy more accessible and attainable for data publishers, since large computing infrastructure is no longer needed to release private information. This directly reduces the economic cost of data privacy, an important consideration for governments and private companies. Our sketching method also allows private functional release to be applied to more real-world datasets, making it easier to release high-dimensional information with differential privacy.

\section*{Appendix}
This section contains proofs and detailed discussion of the theorems. 
\setcounter{theorem}{2}
\setcounter{corollary}{0}

\subsection*{Privacy}
We consider a function $\mathcal{A}: \mathcal{D} \rightarrow \mathbb{R}^{R\times W}$ that takes in a dataset and outputs an $R\times W$ RACE sketch. $\mathcal{A}$ is a randomized function, where the randomness is over the choice of the LSH functions and the (independent) Laplace noise. In this context, the $\epsilon$-differentially private query is the RACE sketch $\mathcal{S}_{\mathcal{D}}$. Our main theorem is that the sketch is $\epsilon$-differentially private. Because differential privacy is robust to post-processing, we can then query $\mathcal{S}_{\mathcal{D}}$ as many times as desired once the sketch is released. 

\paragraph{Proof Sketch:} The sketch $\mathcal{S}_{\mathcal{D}}$ consists of $R$ independent rows. Each row consists of $W$ count values. First, we prove privately release \textit{one} of the rows by adding Laplace noise with variance $\epsilon^{-1}$ (i.e. sensitivity $\Delta = 1$). This is described in Lemma~\ref{lem:row}. Then, we apply Lemma~\ref{lem:row} with $\epsilon / R$ to release \textit{all} the rows, proving the main theorem. 
\begin{lemma}
\label{lem:row}
Consider one row of the RACE sketch, and add independent Laplace noise $\mathrm{Lap}(1/\epsilon)$ to each counter. The row can be released with $\epsilon$-differential privacy. 
\end{lemma}
\begin{proof}
Consider just one of the $W$ counters. It is easy to see that the sensitivity of the counter is 1 because changing a single element from the dataset can only create a change of $\pm1$ in the counter. By the Laplace mechanism, this counter can be released with $\epsilon$-differential privacy by adding Laplace noise $\mathrm{Lap}(1 / \epsilon)$. 

The row contains $W$ counters, and we can release each one with the Laplace mechanism. Therefore, we have $W$ mechanisms $\mathcal{M}_1, ... \mathcal{M}_W$ (one to release each counter) which are independently $\epsilon$-differentially private. The parallel composition theorem~\cite{dwork2014algorithmic} states that if each mechanism is computed on a disjoint subset of the dataset, then we can compute all of the mechanisms (i.e. release all of the counters) with $\epsilon$-differential privacy. 

This is indeed the case for the RACE sketch. Under the LSH function, each element in the dataset maps to exactly one of the $W$ counters. Thus, the $W$ counters are computed on disjoint subsets of the dataset and we can release them with $\epsilon$-differential privacy. 
\end{proof}
\begin{theorem}
For any $R > 0$, $W > 0$, and LSH family $\mathcal{F}$, the output of Algorithm 1, or the RACE sketch $\mathcal{S}_{\mathcal{D}}$, is $\epsilon$-differentially private.
\end{theorem}
\begin{proof}
The sketch is composed of $R$ independent rows. Consider the mechanism from Lemma~\ref{lem:row} to release a single row with $\epsilon$-differential privacy. The sequential composition theorem~\cite{dwork2014algorithmic} states that given $R$ mechanisms $\mathcal{M}_1, ... \mathcal{M}_R$ (one to release each row) which are independently $\epsilon$-differentially private, we can compute all the mechanisms (i.e. release all rows) with $R\epsilon$-differential privacy. 

To construct the sketch, we apply Lemma~\ref{lem:row} to each row with $\epsilon/R$ differential privacy by adding independent Laplace noise $\mathrm{Lap}(R / \epsilon)$ to each counter. The sequential composition theorem guarantees that we can release all $R$ rows with $\epsilon$-differential privacy.
\end{proof}

\subsection*{Utility}

First, we introduce the median-of-means trick in Lemma~\ref{lem:med}. The median-of-means procedure is a standard analysis technique with many applications in the streaming literature. Lemma~\ref{lem:med} is a special case of Theorem 2.1 from~\cite{alon1999space}.
\begin{lemma}
\label{lem:med}
Let $X_1,...X_R$ be $R$ i.i.d. random variables with mean $\mathbb{E}[X] = \mu$ and variance $\leq \sigma^2$. To get the median of means estimate $\hat{\mu}$, break the $R$ random variables into $k$ groups with $m = R/k$ elements in each group.
$$ \hat{\mu} = \mathrm{median}\left(\sum_{i = 1}^m X_i, ..., \sum_{i = (k - 1)m + 1}^{km} X_i\right)$$

Put $k = 8 \log(1/ \delta)$ and $m = R / k$\footnote{For simplicity, we suppose $k$ and $m$ are integers that evenly divide $R$. See~\cite{alon1999space} for the complete analysis.}. Then with probability at least $1 - \delta$, the deviation of the estimate $\hat{\mu}$ from the mean $\mu$ is
$$ |\hat{\mu} - \mu|\leq \sqrt{32\sigma^2 \frac{\log(1/\delta)}{R}}$$
\end{lemma}

Each row $r$ of the sketch provides an estimator $X_r$ to be used in the median-of-means technique. By adding the Laplace noise variance, we can find $\sigma^2$ in Lemma~\ref{lem:med} and provide performance guarantees. 

\begin{theorem}
Let $\hat{f}_{\mathcal{D}}(\mathbf{q})$ be the median-of-means estimate using an $\epsilon$-differentially private RACE sketch with $R$ rows and $\tilde{f}_{\mathcal{D}}(\mathbf{q}) = \sum_{\mathcal{D}} \sqrt{k(\mathbf{x},\mathbf{q})}$. Then with probability $1 - \delta$, 
$$|\hat{f}_{\mathcal{D}}(\mathbf{q}) - f_{\mathcal{D}}(\mathbf{q})| \leq \left(\frac{\tilde{f}^2_\mathcal{D}(\mathbf{q})}{R} + \frac{2}{\epsilon^2}R \right)^{1/2}\sqrt{32 \log 1/\delta}$$

\end{theorem}

\begin{proof}

To use median-of-means, we must bound the variance of the RACE estimate $X_r = \mathcal{S}_{\mathcal{D}}[r,l_r(q)]$ after we add the Laplace noise. Since the Laplace noise is independent of the LSH functions, we simply add the Laplace noise variance $2R^{2}\epsilon^{-2}$ to the variance bound in Theorem 2. 
$$ \mathrm{var}(X_r) \leq \sigma^2 = \left(\tilde{f}_{\mathcal{D}}\right)^2 + 2\frac{R^2}{\epsilon^2}$$
Using this variance bound with Lemma~\ref{lem:med} we have the following statement with probability $1 - \delta$.
$$|\hat{f}_{\mathcal{D}}(\mathbf{q}) - f_{\mathcal{D}}(\mathbf{q})| \leq \left(\frac{\tilde{f}^2_\mathcal{D}(\mathbf{q})}{R} + \frac{2}{\epsilon^2}R \right)^{1/2}\sqrt{32 \log 1/\delta}$$

\end{proof}

\begin{corollary}
Put $R = \lceil\frac{1}{\sqrt{2}}\tilde{f}_\mathcal{D}(\mathbf{q})\epsilon\rceil$. Then the approximation error bound is
$$|\hat{f}_{\mathcal{D}}(\mathbf{q}) - f_{\mathcal{D}}(\mathbf{q})| \leq 16\left(\frac{\tilde{f}_\mathcal{D}(\mathbf{q})}{\epsilon}\log 1/\delta\right)^{1/2} \leq 16\sqrt{\frac{N}{\epsilon}\log 1/\delta}$$
\end{corollary}
\begin{proof}
Take the derivative of $\sqrt{a / R + bR}$ with respect to $R$ to find that $R = \sqrt{a/b}$ minimizes the bound. Put $a = \tilde{f}^2_{\mathcal{D}}(\mathbf{q})$ and $b = 2/\epsilon^2$. The corollary is obtained by substituting $R = \frac{1}{\sqrt{2}}\tilde{f}_{\mathcal{D}}(\mathbf{q})\epsilon$ into Theorem~\ref{thm:tradeoff}. We may replace $\tilde{f}_{\mathcal{D}}(\mathbf{q})$ with $N$ in the inequality because 
$$ \tilde{f}_{\mathcal{D}}(\mathbf{q}) = \sum_{\mathbf{x} \in \mathcal{D}} \sqrt{k(\mathbf{x},\mathbf{q}} \leq \sum_{\mathbf{x}\in\mathcal{D}}1 = N$$

\end{proof}

\bibliographystyle{plain}
\bibliography{main}

\end{document}